\documentclass[runningheads]{llncs}

\usepackage{amssymb}
\usepackage{graphicx}

\usepackage{url}
\urldef{\mailsa}\path|{E-mail: issqdw@mail.sysu.edu.cn (D.Qiu)}|
\urldef{\mailsb}\path||
\newcommand{\keywords}[1]{\par\addvspace\baselineskip
\noindent\keywordname\enspace\ignorespaces#1}

\usepackage{algorithm}
\usepackage{algorithmic}
\usepackage{bbm}

\usepackage{amsmath}
\usepackage{multirow}

\usepackage{graphicx}
\usepackage{epstopdf}
\usepackage{extarrows}
\usepackage{float}
\usepackage[latin1]{inputenc}
\usepackage{tikz}
\usepackage{mathrsfs}
\usepackage{amssymb}
\usepackage{stmaryrd}
\usepackage{mathabx}
\usepackage[normalem]{ulem}
\usepackage{flushend}
\usepackage{amsfonts}
\usetikzlibrary{trees,arrows,automata}
\usepackage{tikz-qtree}
\usepackage{multirow}
\usepackage{tikz}

\usepackage{ntheorem}

\hfuzz=\maxdimen
 \usepackage{indentfirst}
\usepackage{subfigure}
\usepackage{array}
\newcolumntype{C}[1]{>{\centering\arraybackslash}p{#1}}
\usepackage[pdftex, bookmarksnumbered, bookmarksopen, colorlinks, citecolor=blue, linkcolor=blue]{hyperref}

\usepackage[justification=centering]{caption}
\usepackage{caption}
\begin{document}

\mainmatter  

\title{Revisiting Stack Machines and Quantum Stack Machines}

\titlerunning{Stack Machines }

\author{Daowen Qiu}
\authorrunning{Daowen Qiu}
\institute{School of Computer Science and Engineering, Sun Yat-sen
University, Guangzhou 510006, China\\
\mailsa\\}

\toctitle{Simulations of Quantum Turing Machines } \tocauthor{
Daowen Qiu} \maketitle
\begin{abstract}

Multi-stack machines and Turing machines can simulate to each other. In this paper, we give a succinct definition of multi-stack machines, and from this definition it is clearly seen that pushdown automata and deterministic finite automata are special cases of multi-stack machines. Also, with this mode of definition, pushdown automata and deterministic pushdown automata are equivalent and recognize all context-free languages. In addition, we are motivated to reformulate concise definitions of quantum pushdown automata and quantum multi-stack machines.

\end{abstract}

\section{Introduction}
In the theory of classical computation \cite{Hop79}, both multi-stack
machines, as a generalization of pushdown automata, and multi-counter
machines can efficiently simulate Turing machines (two-stack
machines and two-counter
machines can efficiently simulate Turing machines.)
\cite{Min61,Fi66,Hop79}.  

A different definition of pushdown automata (named as {\it pushdown automata-II}) was given in \cite{Watrous}, and in this note we will further show this definition is equivalent to the traditional pushdown automata \cite{Hop79}.   One-way nondeterministic stack automata were studied from the point of view of space complexity \cite{IJMP2020}. Quantum pushdown automata and quantum stack machines as well as quantum counter machines were studied in \cite{Go00,Qiu05,Nak14,Kra99,YKI05}. However, these quantum computing models seem somewhat complex in practice from the viewpoint of physical realizability. 

In fact, the definition of pushdown automata-II is non-deterministic and allows empty transition as nondeterministic finite automata (NFA). If it is required that pushdown automata-II are deterministic without empty transition, then deterministic pushdown automata-II can be defined as deterministic finite automata (DFA).  Actually, in this paper, pushdown automata-II and deterministic pushdown automata-II are proved to be equivalent.

According to the mode of definition for pushdown automata-II, we can further define multi-stack machines by adding stacks to deterministic pushdown automata-II, and a number of examples are given to show the ability of these new computing models. In particular, from these definitions it is clearly seen that DFA and deterministic  pushdown automata-II as well as two-stack machines are increasingly stronger.

Quantum pushdown automata have been studied in \cite{MC00,Go00,Nak14,QY04}, but these definitions are complex in a way due to the well-formedness conditions involved in these definitions. Motivated by this new definitions of pushdown automata, we may define new quantum pushdown automata with succinct evolutions, as quantum finite automata (e.g., \cite{AY15,Qiu25}). Also, we can further re-define quantum multi-stack machines that are conciser than the definition in \cite{Qiu05} to  certain extent.

The remainder of the paper is organized as follows. In Section 2, we present a new definition of multi-stack machines and it is proved that two non-context-free languages are recognized by such multi-stack machines. Then in Section 3, new deterministic pushdown automata and nondeterministic pushdown automata are defined, and the equivalence to traditional pushdown automata is proved. In Section 4, new definitions of quantum pushdown automata and quantum multi-stack machines are presented. Finally, conclusions are made in Section 5.

\section{New multi-stack machines}

Multi-stack machines, as a generalization of pushdown automata, have been described generally (e.g.,  \cite{Hop79}). 
Here we give a new  definition of two-stack machines (multi-stack machines are similar).  First, we need some notations. Let $\Gamma$ denote a finite set of stacks, and let $\varepsilon$ denote an empty string. Denote $\Gamma_i(\downarrow)=\{X(i,\downarrow): X\in \Gamma\}$ and 
$\Gamma_i(\uparrow)=\{X(i,\uparrow): X\in \Gamma\}$, $i=1,2$, where $X(i,\downarrow)$ and $X(i,\uparrow)$ denote ``push" $X$ and ``pop" $X$ in stack $i$, respectively. In addition, we denote $\Gamma_i(\updownarrow)=\Gamma_i(\downarrow)\cup \Gamma_i(\uparrow)$, and
\begin{equation}
\Gamma_{1,2}(\updownarrow)=(\Gamma_1(\updownarrow)\times \Gamma_2(\updownarrow))\cup (\Gamma_1(\updownarrow)\times \{\varepsilon\})\cup (\{\varepsilon\}\times \Gamma_2(\updownarrow)).
\end{equation}
More exactly, we can represent $\Gamma_{1,2}(\updownarrow)$ as follows.
\begin{eqnarray}
\Gamma_{1,2}(\updownarrow)&=&\{(A,B): A\in  \Gamma_1(\updownarrow),B\in  \Gamma_2(\updownarrow)\} \nonumber \\
&&\cup \{(A,\varepsilon):A\in \Gamma_1(\updownarrow)\}\cup \{(\varepsilon, B):B\in \Gamma_2(\updownarrow)\}.
\end{eqnarray}

Without the number of stack, $\Gamma(\downarrow)=\{X(\downarrow): X\in \Gamma\}$ and 
$\Gamma(\uparrow)=\{X(\uparrow): X\in \Gamma\}$, where $X(\downarrow)$ and $X(\uparrow)$ denote ``push" $X$ and ``pop" $X$, respectively.

\begin{definition}
Let $\Gamma$ be a set of stack.  A string $x\in \Gamma_i(\updownarrow)^*$ ($i=1,2$) is called to be valid in stack $i$ if its changing from the most left to the most right in $x$ leads to the $i$th stack being empty for prior empty stack $i$. A string $(A_1,B_1)(A_2,B_2)\ldots (A_k,B_k)\in \Gamma_{1,2}(\updownarrow)^*$ is called to be valid if both $A_1A_2\ldots A_k$ and $B_1B_2\ldots B_k$ are valid in stacks $1$ and $2$, respectively.

\end{definition}

For example, for $X,Y\in \Gamma$, then both $X(1,\downarrow)Y(1,\downarrow)X(1,\downarrow)X(1,\uparrow)Y(1,\uparrow)X(1,\uparrow)$ and $Y(2,\downarrow)X(2,\downarrow)X(2,\uparrow)Y(2,\downarrow)Y(2,\uparrow)X(2,\uparrow)$ are valid, and therefore
\begin{eqnarray}
&&(X(1,\downarrow), Y(2,\downarrow))(Y(1,\downarrow),X(2,\downarrow))(X(1,\downarrow),X(2,\uparrow))(X(1,\uparrow),Y(2,\downarrow))\nonumber\\
&&(Y(1,\uparrow),Y(2,\uparrow))(X(1,\uparrow),X(2,\uparrow))
\end{eqnarray}
is valid. Clearly, 
$X(1,\downarrow)Y(1,\downarrow)X(1,\downarrow)Y(1,\uparrow)Y(1,\uparrow)X(1,\uparrow)$ is invalid.

\begin{definition}
A  two-stack machine is defined
as $M=(Q,\Sigma,\Gamma, \Delta, \delta, q_{0}, F)$ where
$Q$ is a finite set of states, $\Sigma$ is a finite input alphabet, $\Gamma$
is the stack alphabet, $\Delta$ is a finite set of tape symbols, $q_{0}\in Q$ is the initial
state, and $F\subseteq Q$ denotes the set of accepting states, and transition function $\delta$ is defined as
follows:
\begin{equation}
\delta:
Q\times(\Gamma_{1,2}(\updownarrow)\cup \Sigma\cup\Delta)\rightarrow Q.
\end{equation}

\end{definition}

As deterministic finite automata, it is clear that $M$ outputs a state in $Q$ for any inputting string $x\in(\Gamma_{1,2}(\updownarrow)\cup \Sigma\cup\Delta)^*$ with the initial state $q_0$, and this outputted state is represented by $\delta^*(q_0,x)$, as that in deterministic finite automata.

 We still need some notations. Let $S_1$ and $S_2$ be two sets. Then, for any string $x\in (S_1\cup S_2)^*$, $[x]_{S_1}$ denotes the string that is obtained by deleting all elements not in $S_1$ from $x$, that is to say, $[x]_{S_1}$ is a string by only keeping the elements in $S_1$ from $x$.
 
 The computing procedure of this two-stack machine $M$ can be formally defined and described as follows.
 
 \begin{definition} Given a finite alphabet $\Sigma$, and for any input string $x=a_1a_2\cdots a_k\in\Sigma^{*}$, $x$ is called to be accepted  by $M=(Q,\Sigma,\Gamma, \Delta, \delta, q_{0}, F)$, if there is a string $s\in (\Gamma_{1,2}(\updownarrow)\cup \Sigma\cup\Delta)^*$ satisfying the following conditions:
\begin{enumerate}
\item $[s]_{\Sigma}=x$.
\item $\delta^*(q_0,s)\in F$.
\item $[s]_{\Gamma_{1,2}(\updownarrow)}$ is a valid string of stacks changing.
\end{enumerate}

Otherwise,  $x$ is called to be rejected  by $M$.

\end{definition}

A language $L\subseteq\Sigma^*$ is called to be recognized by a two-stack machine $M$, if for any $x\in L$, $x$ is accepted by $M$, and for any $x\in \Sigma^*\setminus L$, $x$ is rejected by $M$.

Next we present two non-context free languages that can be recognized by two-stack machines.

\begin{example}
The language $L_{eq}=\{0^n1^n2^n: n\in N\}$ can be recognized by a two-stack machine $M(L_{eq})=(Q,\Sigma,\Gamma, \Delta, \delta, q_{0}, F)$ illustrated by the following figure of transformation, where 
\begin{itemize}
\item $Q=\{q_0,q_1,q_2,q_3,q_1^{'},q_2^{'},q_3^{'},q_a\}$;
\item $\Sigma=\{0,1,2\}$;
\item $\Gamma=\{Z_0,X\}$;
\item  $\Delta=\{t_0\}$, $q_0$ and $q_a$ ($F=\{q_a\}$) are initial and accepting states, respectively,
\end{itemize}
and $\delta$ is illustrated in Fig.1 as follows.

\begin{figure*}[htbp]
		\centering
		\includegraphics[width=\linewidth]{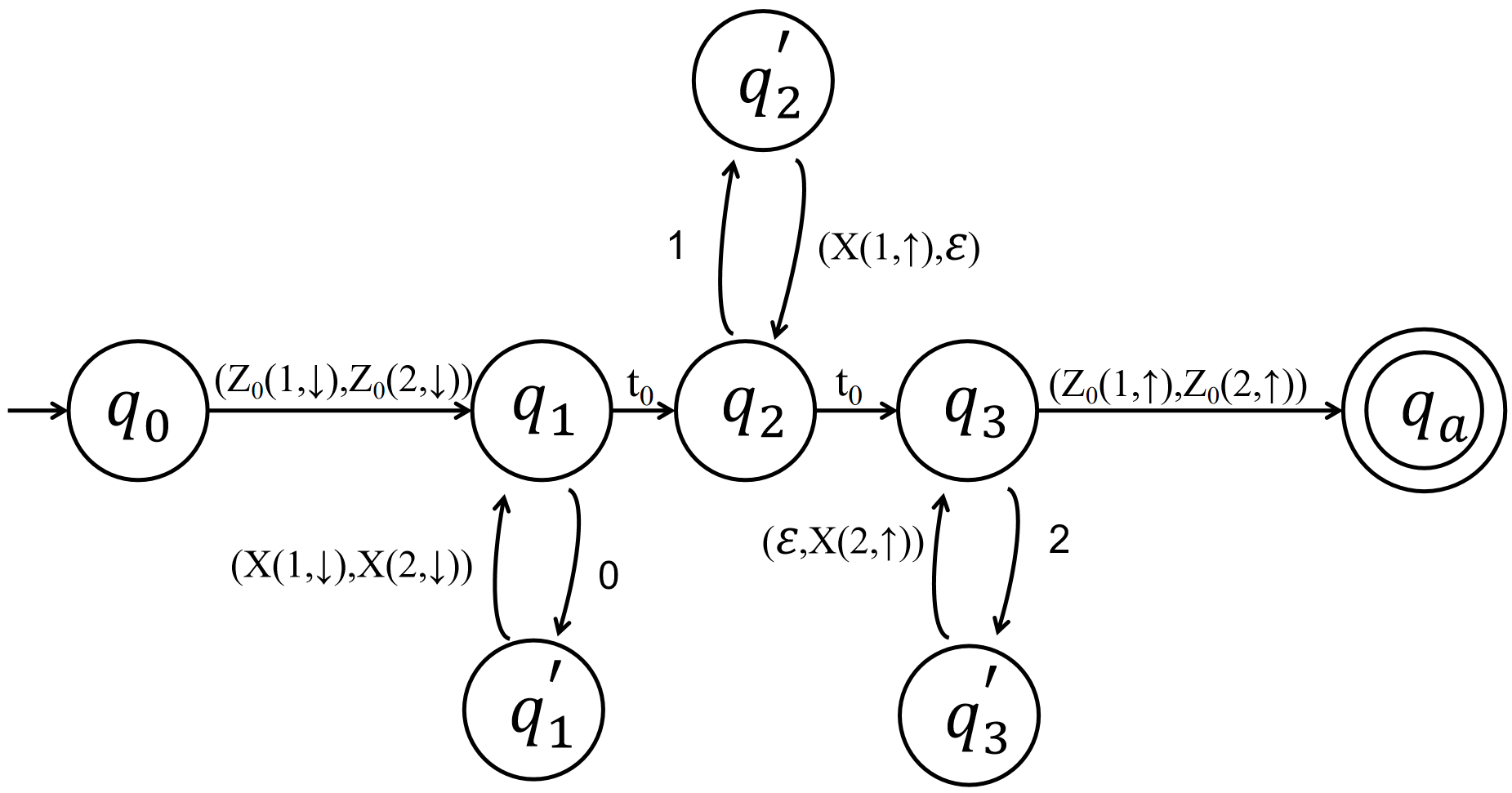}
		\setlength{\abovecaptionskip}{-0.01cm}
		\caption{Two-stack machine $M(L_{eq})$ recognizes $L_{eq}$.}
		\label{1}
	\end{figure*}
	
\end{example}	
	
	
	We give another example in the following.
	
	\begin{example}	

The language $L_{w}=\{w\#w: w\in \{0,1\}^*\}$ can be recognized by a two-stack machine $M(L_{w})=(Q,\Sigma,\Gamma, \Delta, \delta, q_{0}, F)$ illustrated by the following figure of transformation, where 
\begin{itemize}
\item $Q=\{q_0,q_1,q_2,q_3,q_1^{'},q_1^{''}, q_2^{'},q_2^{''},q_3^{'},q_3^{''}q_a\}$;
\item $\Sigma=\{0,1,\#\}$;
\item $\Gamma=\{Z_0,0,1\}$;
\item  $\Delta=\{t_0\}$, $q_0$ and $q_a$ ($F=\{q_a\}$) are initial and accepting states, respectively,
\end{itemize}
and $\delta$ is illustrated in Fig.2 as follows.

\begin{figure*}[htbp]
		\centering
		\includegraphics[width=\linewidth]{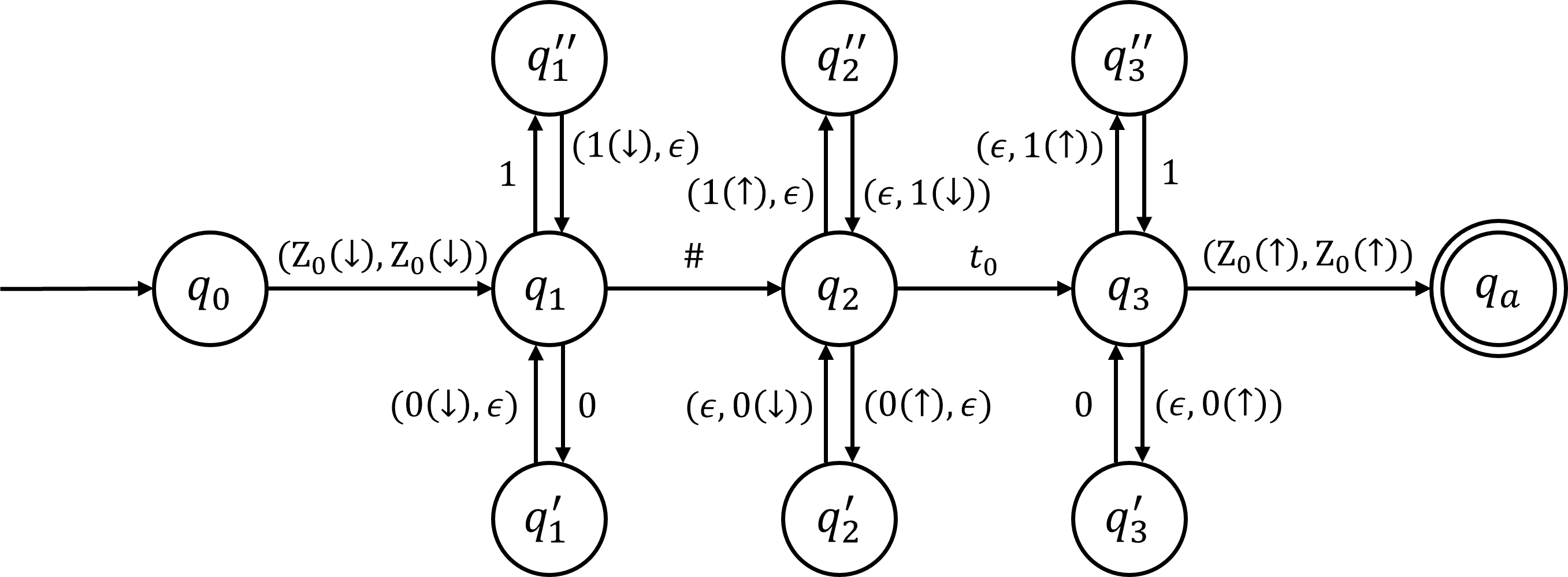}
		\setlength{\abovecaptionskip}{-0.01cm}
		\caption{Two-stack machine $M(L_{w})$ recognizes $L_{w}$.}
		\label{2}
	\end{figure*}

\end{example}
\newpage

\section{New pushdown automata}

Next we show the equivalence between the traditional pushdown automata in \cite{Hop79} and the  redefined pushdown automata in \cite{Watrous} that can be thought of as one-stack machines without tape symbols defined above.
First let us recall the traditional pushdown automata in \cite{Hop79}.

\begin{definition}
\cite{Hop79} A pushdown automaton  can be defined as 
\[
M = (Q, \Sigma, \Gamma, \delta, q_0, Z_0, F)
\]
where
\begin{itemize}
    \item \( Q \) is finite set of states;
    \item \( \Sigma \) is finite input alphabet;
    \item \( \Gamma \) is finite set of stacks;
    \item \( \delta: Q \times (\Sigma \cup \{\varepsilon\}) \times \Gamma \to 2^{Q \times \Gamma^*} \) is a transform function;
    \item $q_0\in Q$ is an initial state;
    \item \( Z_0\in \Gamma \) is an initial stack symbol;
    \item \( F\subseteq Q \) is a set of accepting states.
\end{itemize}

\end{definition}

The above model is named as  pushdown automata-I (PDA-I) that recognize all context-free languages \cite{Hop79}.

\begin{lemma}\cite{Hop79}
Let $\Sigma$ be a finite alphabet. For any language $L\subseteq \Sigma^*$, then $L$ is a context-free language if and only of $L$ is recognized by a pushdown automaton-I.

\end{lemma}

Another type of pushdown automata given in \cite{Watrous} is called as pushdown automata-II (PDA-II) and defined as follows:
\begin{definition}
A pushdown automaton-II (PDA-II)  can be defined as
 \[
M= (Q, \Sigma, \Gamma, \delta, q_0, F)
\]
where
\begin{itemize}
   \item \( Q \) is finite set of states;
    \item \( \Sigma \) is finite input alphabet;
    \item \( \Gamma \) is finite set of stacks;
    \item \( \delta: Q \times (\Sigma \cup \{\varepsilon\} \cup  \Gamma(\updownarrow)) \to 2^{Q} \) a transform function, where $\Gamma(\updownarrow)=\{X(\downarrow):X\in\Gamma\}\cup\{X(\uparrow):X\in\Gamma\}$;
    \item $q_0\in Q$ is an initial state;
    \item \( F\subseteq Q \) is a set of accepting states.
\end{itemize}

\end{definition}

We describe the computing procedure of PDA-II that  is   formally defined  as follows.
 
 \begin{definition} Given a finite alphabet $\Sigma$, and for any input string $x=a_1a_2\cdots a_k\in\Sigma^{*}$, $x$ is called to be accepted  by the PDA-II $M=(Q,\Sigma,\Gamma, \delta, q_{0}, F)$, if there is a string $s\in (\Gamma(\updownarrow)\cup \Sigma\cup\{\varepsilon\})^*$ satisfying the following conditions:
\begin{enumerate}
\item $[s]_{\Sigma}=x$.
\item $\delta^*(q_0,s)\cap F\not=\emptyset$.
\item $[s]_{\Gamma(\updownarrow)}$ is a valid string of stacks changing.
\end{enumerate}

Otherwise, $x$ is called to be rejected  by $M$.

\end{definition}

The following lemma comes from \cite{Watrous}. 

\begin{lemma}\cite{Watrous}
Let $\Sigma$ be a finite alphabet. For any language $L\subseteq \Sigma^*$, then $L$ is a context-free language if and only of $L$ is recognized by a pushdown automaton-II.

\end{lemma}

From the above two lemmas it follows that PDA-I and PDA-II are equivalent, that is, any language recognized by PDA-I can be also recognized by PDA-II, and vice versa.

\begin{theorem}
Let $\Sigma$ be a finite alphabet. Then for any PDA-I $M$ over $\Sigma$, there is a PDA-II equivalent to PDA-I (recognizing the same language), and vice versa.

\end{theorem}
\begin{proof}
It follows from Lemmas 1 and 2.  In fact, we can directly prove  that PDA-I and PDA-II can simulate to each other by means of a constructivity method. Here we only outline the basic ideas without presenting the details. 

Given a PDA-I
\[
M = (Q, \Sigma, \Gamma, \delta,  q_0, Z_0, F),
\]
we construct a PDA-II as
\[
M_N=(Q_N, \Sigma, \Gamma, \mu, q_N, F)
\]
where $Q_N=Q\cup\{q_N\}\cup Q'$ for a set of some auxiliary states $Q'$ to be fixed. For any transformation in PDA
\[
(p, \gamma)\in\delta(q,a,X),
\]
where $q, p\in Q$, $a\in \Sigma\cup\{\varepsilon\}$, $X\in\Gamma$,  $\gamma\in\Gamma^*$ and $\gamma=X_1X_2...X_n$ with $X_i\in\Gamma$ ($i=1,2,\ldots,n$), we define the corresponding transformations in PDA-II 
\[
{q_p}^{(0)}\in\mu(q, a),
\]
\[
{q_p}^{(1)}\in\mu({q_p}^{(0)}, X( \uparrow )),
\]
\[
{q_p}^{(2)}\in\mu({q_p}^{(1)}, X_n(\downarrow )),
\]
\[
...
\]
\[
{q_p}^{(n)}\in\mu({q_p}^{(n-1)}, X_2 (\downarrow )),
\]
\[
p\in\mu({q_p}^{(n)}, X_1 (\downarrow )),
\]
where ${q_p}^{(0)}, {q_p}^{(1)},...,{q_p}^{(n)}\in Q'$. In addition, we define
\[
q_0\in\mu(q_N,Z_0(\downarrow)).
\]
Then it can be checked that the above $M$ and $M_N$ recognize the same language. 

On the contrary, for  a given PDA-II $M = (Q, \Sigma, \Gamma, \delta, q_0, F)$, we define a PDA-I as  $M_N= (Q, \Sigma, \Gamma', \delta', q_0, Z_0, F)$, where $\Gamma'=\Gamma \cup \{Z_{0}\}$,  and $\delta'$ is defined as follows: 
\begin{enumerate}

\item For  any $ a_{k+1} \in \Sigma\cup\{\varepsilon\}$,  $ r_{k+1} \in \delta(r_{k}, a_{k+1})$ and   $X \in \Gamma \cup \{Z_0\}$,
   \[
   (r_{k+1}, X) \in \delta'(r_{k}, a_{k+1}, X).
   \]

\item For any $a_{k+1} \in  \Gamma(\updownarrow)$, and $r_{k+1} \in \delta(r_{k}, a_{k+1})$,  

(1) if $a_{k+1} = Z_{k+1}(\uparrow )$, then
    \[
    (r_{k+1}, \epsilon) \in \delta'(r_{k}, \epsilon, Z_{k+1});
    \]

(2) if $a_{k+1} =  Z_{k+1}(\downarrow)$, then for any $X \in \Gamma$   
 \[
    (r_{k+1}, Z_{k+1}X) \in \delta'(r_{k}, \epsilon, X).
    \]

Also it can be checked that the above $M$ and $M_N$ recognize the same language.  So, the proof is completed. \quad $\blacksquare$
\end{enumerate}
\end{proof}

\begin{remark}
In view of the definitions of PDA-I and PDA-II, we know both models (more exactly, their transformation functions) are {\it nondeterministic}. Concerning PDA-I, it has been proved that PDA-I and deterministic PDA-I are not equivalent \cite{Hop79}, since the language $L_{wwr}=\{ww^R: w\in\ \{0,1\}^*\}$ can be recognized by PDA-I, but it can not be recognized by any deterministic PDA-I \cite{Hop79}.
\end{remark}

Now we give the definition of deterministic  PDA-II formally as follows.

\begin{definition}
A deterministic pushdown automaton-II (DPDA-II)  can be defined as
\[
M= (Q, \Sigma, \Gamma, \delta, q_0, F)
\]
where
\begin{itemize}
   \item \( Q \) is finite set of states;
    \item \( \Sigma \) is finite input alphabet;
    \item \( \Gamma \) is finite set of stacks;
    \item \( \delta: Q \times (\Sigma \cup   \Gamma(\updownarrow)\}) \to Q \) a transform function, where $\Gamma(\updownarrow)=\{X(\downarrow):X\in\Gamma\}\cup\{X(\uparrow):X\in\Gamma\}$;
    \item $q_0\in Q$ is an initial state;
    \item \( F\subseteq Q \) is a set of accepting states.
\end{itemize}

\end{definition}

Clearly, if $\Sigma\cup\Gamma(\updownarrow)$ is thought of as the input alphabet, then deterministic PDA-II and PDA-II $M$ are  DFA and nondeterministic finite automata (NFA), respectively.

By using the subset construction method that can construct a DFA  being equivalent to any given NFA \cite{Hop79}, it is clear that for any PDA-II $M$, there is a deterministic PDA-II $M^{'}$ equivalent to $M$. 
So, we have the following result.
\begin{theorem} Let $\Sigma$ be a finite alphabet. For any PDA-II $M=(Q, \Sigma, \Gamma, \delta, q_0, F)$,  there is a DPDA-II $M^{'}=(Q^{'}, \Sigma, \Gamma, \delta^{'}, q_0^{'}, F^{'})$ that is equivalent to $M$.

\end{theorem}

\begin{proof}

First, we regard $\Sigma \cup   \Gamma(\updownarrow)$ as the same input alphabet of $M$ and $M^{'}$. Then $M^{'}$ follows from the subset construction method that transfers  a nondeterministic finite automaton to deterministic finite automaton  \cite{Hop79}.

If  $\Sigma \cup   \Gamma(\updownarrow)$ are regarded as the input alphabet of $M$ and $M^{'}$, then we denote $L(\Sigma \cup  \Gamma(\updownarrow)) \subseteq (\Sigma \cup  \Gamma(\updownarrow))^*$ as the language over $\Sigma \cup  \Gamma(\updownarrow)$  recognized by $M$ and $M^{'}$, and clearly it is a regular language over $\Sigma \cup  \Gamma(\updownarrow)$.
The rest is to show that $M$ and $M^{'}$ also recognize the same language over $\Sigma$.

Let $L(\Sigma) $ and $L(\Sigma)^{'} $ denote the languages recognized by $M$ and $M^{'}$,
respectively. Our aim is to show $L(\Sigma) =L(\Sigma)^{'} $.

As is shown \cite{Watrous}, the language denoted by $L(\Gamma(\updownarrow))$, consisting of all valid strings of stacks in $ \Gamma(\updownarrow)^*$,  is context-free. Then the extended language, denoted by $L(\Gamma(\updownarrow),\Sigma)$, which is composed by adding any elements from $\Sigma^*$ to any position in the strings of $L(\Gamma(\updownarrow))$, is also context-free. Formally,  $L(\Gamma(\updownarrow),\Sigma)$ can be described mathematically as follows.
\begin{equation}
L(\Gamma(\updownarrow),\Sigma)=\{w:w\in (\Gamma(\updownarrow)\cup\Sigma)^{*}, [w]_{\Gamma(\updownarrow)}\in L(\Gamma(\updownarrow))\}.
\end{equation}
Now we have 
\begin{eqnarray}
L(\Sigma)&=&\left[L(\Sigma \cup  \Gamma(\updownarrow))\cap  L(\Gamma(\updownarrow),\Sigma)\right]_{\Sigma} \nonumber\\
&=&L(\Sigma)^{'},
\end{eqnarray}
and we have completed the proof.  \quad $\blacksquare$
\end{proof}

For a deterministic PDA-II (DPDA-II) $M$ with the input alphabet $\Sigma$ and the set of stack symbols $\Gamma$, if $\Sigma\cup\Gamma(\updownarrow)$ is thought of as the input alphabet, then $M$ is a DFA. From the above proof  a corollary follows.

\begin{corollary}

For any (deterministic) PDA-II $M$ with the input alphabet $\Sigma$ and the set of stack symbols $\Gamma$, let $L(\Sigma\cup\Gamma(\updownarrow))$ denote the language recognized by $M$ with the input alphabet  $\Sigma\cup\Gamma(\updownarrow)$, then
the language recognized by $M$ is $\left[L(\Sigma \cup  \Gamma(\updownarrow))\cap  L(\Gamma(\updownarrow),\Sigma)\right]_{\Sigma}$.

\end{corollary}

We give two examples  to show that the language $L_{wwr}=\{ww^R: w\in \{0,1\}^*\}$ can be recognized by two PDA-II and PDA-II (though it is not recognized by DPDA-I \cite{Hop79}). First, we present a  PDA-II as follows to recognize  $L_{wwr}$.

\begin{example}	

The language $L_{wwr}=\{ww^R: w\in \{0,1\}^*\}$ can be recognized by a PDA-II $M_1(L_{wwr})=(Q,\Sigma,\Gamma,\delta,q_0,F)$ illustrated by the following figure of transformations, where 
\begin{itemize}
\item $Q=\{q_0,q_1,q_2,q_1^{'},q_1^{''}, q_2^{'},q_2^{''},q_a\}$;
\item $\Sigma=\{0,1\}$;
\item $\Gamma=\{Z_0,X_0,X_1\}$;
\item  $q_0\in Q$ is an initial state, and $q_a$ is the accepting state,
\end{itemize}
and $\delta$ is illustrated in Fig.3 as follows.

\begin{figure*}[htbp]
		\centering
		\includegraphics[width=\linewidth]{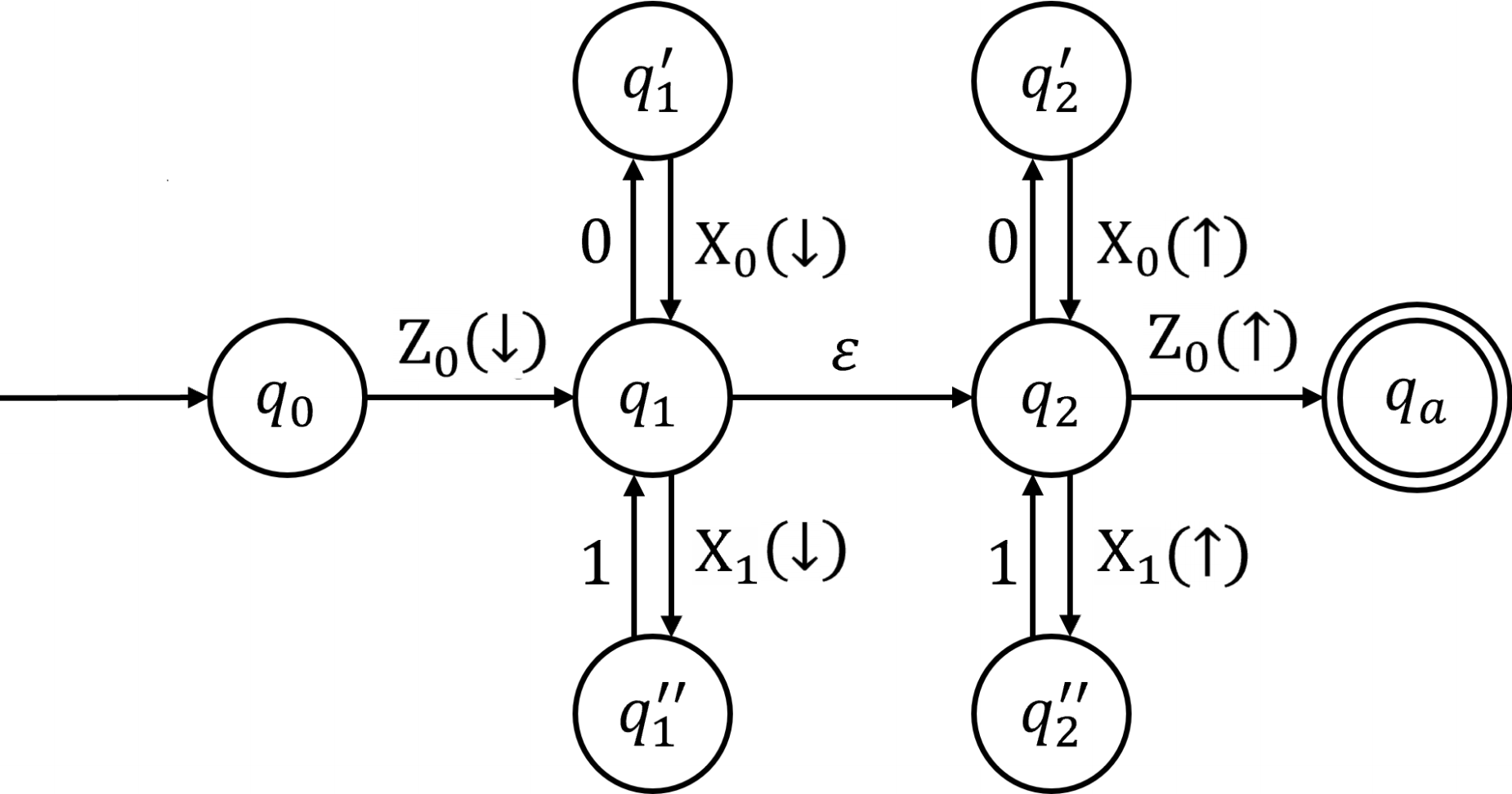}
		\setlength{\abovecaptionskip}{-0.01cm}
		\caption{PDA-II $M_1(L_{wwr})$ recognizes $L_{wwr}$.}
		\label{3}
	\end{figure*}

\end{example}






\begin{remark}

It is easy to construct a DPDA-II to recognize $L_{wwr}$ in terms of $M_2(L_{wwr})$ by adding some states and transformations.

\end{remark}

Here we would like to point out two-stack machines are more general models, and DPDA-II as well as DFA are special two-stack machines.

\begin{remark}

Let $M=(Q,\Sigma,\Gamma, \Delta, \delta, q_{0}, F)$  be a two-stack machine. Then we have:

\begin{itemize}
   \item if only one stack is used and tape symbols are removed, then $M$ reduces to a DPDA-II;
    \item if two stacks and tape symbols are removed, then $M$ reduces to a DFA.
       \end{itemize}

\end{remark}

\section{Quantum pushdown automata and quantum multi-stack machines}

In this section, we would consider quantum version of these models.
In light of the previous definitions, we can define quantum pushdown automata. To avoid confusion with the definitions of quantum pushdown automata in  \cite{Go00,QY04,Nak14}, we name it as quantum pushdown automata-II (QPDA-II). We suppose some background concerning quantum automata is known and we can refer to \cite{AY15,Qiu25}.

Suppose $Q$ is finite set of states and denote $|Q|$ as the number of states in $Q$. Then let ${\cal H}_Q$ be the Hilbert space with  basis  identified with $Q$.

\begin{definition}
A quantum pushdown automaton-II (QPDA-II)  can be defined as
 \[
M= (Q, \Sigma, \Gamma, {\cal U}, q_0, F)
\]
where
\begin{itemize}
   \item \( Q \) is a finite set of states;
    \item \( \Sigma \) is a finite input alphabet;
    \item \( \Gamma \) is a finite set of stacks;
    \item ${\cal U}=\{U_{t}: t\in\Sigma\cup \Gamma(\updownarrow)\}$, where $U_t$ are unitary operators on Hilbert space ${\cal H}_Q$;
    \item $|q_0\rangle$ is an initial state, where $q_0\in Q$;
    \item \( F\subseteq Q \) is a set of accepting states.
\end{itemize}

\end{definition}

For any $x\in \Sigma^*$, denote 
\begin{equation}
 S_x=\{s\in (\Gamma(\updownarrow)\cup \Sigma)^*: [s]_{\Sigma}=x,  [s]_{\Gamma(\updownarrow)}\in Val(\Gamma(\updownarrow))\},
 \end{equation}
 where $Val(\Gamma(\updownarrow))$ denotes the set consisting of all valid strings of stacks changing.

As usual, the accepting probability of QPDA-II $M$ is defined as follows.  For any $x\in \Sigma^*$, the probability $P_a(x)$ of $x$ being accepted by $M$ is 
\begin{equation}
P_a(x)=\sup\{\sum_{q\in F}\langle q|U_s|q_0\rangle: s\in S_x\}
\end{equation}
where $U_s=U_{a_k}U_{a_{k-1}}\cdots U_{a_1}$ if $s=a_1a_2\cdots a_k$.


In \cite{Qiu05}, quantum stack machines are defined. Here we give a different definition, named as quantum two-stack machines-II (Q2SM-II).

\begin{definition}
A  quantum two-stack machine is defined
as $$M=(Q,\Sigma,\Gamma, \Delta, \delta, q_{0}, F)$$ 
where
\begin{itemize}
   \item \( Q \) is a finite set of states;
    \item \( \Sigma \) is a finite input alphabet;
    \item \( \Gamma \) is a finite set of stacks;
     \item \( \Delta \) is a finite set of tape symbols;
    \item ${\cal U}=\{U_{t}: t\in\Sigma\cup \Gamma(\updownarrow)\cup \Delta\}$, where $U_t$ are unitary operators on Hilbert space ${\cal H}_Q$;
    \item $|q_0\rangle$ is an initial state, where $q_0\in Q$;
    \item \( F\subseteq Q \) is a set of accepting states.
\end{itemize}
\end{definition}

We can also define the accepting probability for Q2SM-II $M$. For any $x\in \Sigma^*$, denote 
\begin{equation}
 S_x^{'}=\{s\in (\Gamma_{1,2}(\updownarrow)\cup \Sigma\cup \Delta)^*: [s]_{\Sigma}=x,  [s]_{\Gamma_{1,2}(\updownarrow)}\in Val(\Gamma_{1,2}(\updownarrow))\},
 \end{equation}
 where $Val(\Gamma_{1,2}(\updownarrow))$ denotes the set consisting of all valid strings of stacks changing.

The accepting probability of Q2SM-II $M$ is defined as follows.  For any $x\in \Sigma^*$, the probability $P_a(x)$ of $x$ being accepted by $M$ is 
\begin{equation}
P_a(x)=\sup\{\sum_{q\in F}\langle q|U_s|q_0\rangle: s\in S_x^{'}\}
\end{equation}
where $U_s=U_{a_k}U_{a_{k-1}}\cdots U_{a_1}$ if $s=a_1a_2\cdots a_k$.

It is seen that Q2SM-II are more general quantum computing models, since both QPDA-II and quantum finite automata (QFA) \cite{MC00} are special cases of Q2SM-II. Also, QFA are special case of QPDA-II.

The relationships between Q2SM-II and quantum Turing machines (QTMs)  \cite{BV97,De85} need to be further considered. Also, how to do simulations by quantum circuits is  to be studied furthermore, as QTMs simulated by quantum circuits \cite{Yao}.

\section{Conclusions}

In this paper, we have revisited a number of classical computing models, and motivated by these models, we have defined corresponding quantum computing models. The main results include: (1)  We have defined deterministic pushdown automata and non-deterministic pushdown automata, and  the equivalence to the traditional nondeterministic  pushdown automata \cite{Hop79} has been proved, while the traditional nondeterministic  pushdown automata and deterministic pushdown automata \cite{Hop79} are not equivalent. (2) We have re-defined multi-stack machines by means of adding stacks and tape symbols in deterministic pushdown automata, and showed that some non-context-free languages can be recognized by this new defined multi-stack machines. (3) Inspired by these models, we have defined new quantum pushdown automata and quantum multi-stack machines.

From theses re-defined computing models, we can clearly see the relationships between deterministic finite automata (DFA), deterministic pushdown automata, and multi-stack machines, and it is seen that  their structures are gradually more and more complex, while their ability of recognizing languages is stronger and stronger. 

It is hoped that succinct quantum computing models with strong ability of computing can be designed and this may provide certain insights for realizing hardware of  quantum computers in practice. In a way, quantum pushdown automata  and quantum multi-stack machines defined in this paper are concise in comparison to  the existing models  \cite{Go00,QY04,Nak14,Qiu05}, but there are still pending problems to be further clarified: (1) Simulation of Turing machines by the new multi-stack machines is required to be carefully considered. (2) 
How about the ability of QPD-II compared to other quantum pushdown automata for recognizing languages? (3) How about the ability of Q2SM-II  compared to quantum Turing machines?


\end{document}